\documentclass{llncs}

\usepackage{tikz}
\usetikzlibrary{shapes,arrows}
\usepackage{amsmath}
\usepackage{amssymb}
\usepackage{amstext}
\usepackage[vlined,boxed]{algorithm2e}
\usepackage{alltt}
\usepackage{wrapfig}

\tikzstyle{line} = [draw, color=black, -latex]

\begin{document}

\title{Ranking Functions for Vector Addition Systems}

\author{
Florian Zuleger\\
zuleger@forsyte.at
}
\institute{TU Wien}
\maketitle

\begin{abstract}
Vector addition systems are an important model in theoretical computer science and have been used for the analysis of systems in a variety of areas.
Termination is a crucial property of vector addition systems and has received considerable interest in the literature.
In this paper we give a complete method for the construction of ranking functions for vector addition systems with states.
The interest in ranking functions is motivated by the fact that ranking functions provide valuable additional information in case of termination:
They provide an explanation for the progress of the vector addition system, which can be reported to the user of a verification tool, and can be used as certificates for termination.
Moreover, we show how ranking functions can be used for the computational complexity analysis of vector addition systems (here complexity refers to the number of steps the vector addition system under analysis can take in terms of the given initial vector).
\end{abstract}

\newcommand{\length}{\ensuremath{u}}
\newcommand{\field}{\ensuremath{\mathbb{S}}}
\newcommand{\vass}{\ensuremath{\mathcal{V}}}
\newcommand{\vars}{\ensuremath{\mathit{Var}}}
\newcommand{\var}{\ensuremath{x}}
\newcommand{\loc}{\ensuremath{l}}
\newcommand{\locations}{\mathit{Locs}}
\newcommand{\transitions}{\mathit{Trns}}
\newcommand{\edges}{\ensuremath{E}}
\newcommand{\update}{\ensuremath{d}}
\newcommand{\updates}{\ensuremath{D}}
\newcommand{\paath}{\ensuremath{\pi}}
\newcommand{\val}{\ensuremath{\nu}}
\newcommand{\states}{\ensuremath{\mathit{St}}}
\newcommand{\state}{\ensuremath{\sigma}}
\newcommand{\dimension}{\ensuremath{n}}
\newcommand{\dimOp}{\ensuremath{\mathit{dim}}}
\newcommand{\Val}{\ensuremath{\mathit{Val}}}
\newcommand{\lex}{\ensuremath{\mathit{lex}}}
\newcommand{\cycle}{\ensuremath{C}}
\newcommand{\multicycle}{\ensuremath{M}}
\newcommand{\zeroVec}{\ensuremath{\mathbf{0}}}
\newcommand{\oneVec}{\ensuremath{\mathbf{1}}}
\newcommand{\parameter}{\ensuremath{N}}
\newcommand{\valueSum}{\ensuremath{\mathit{val}}}
\newcommand{\counters}{\ensuremath{\mu}}
\newcommand{\countersAlt}{\ensuremath{\rho}}
\newcommand{\transition}{\ensuremath{t}}
\newcommand{\flowMatrix}{\ensuremath{F}}
\newcommand{\rankCoeff}{\ensuremath{r}}
\newcommand{\offsets}{\ensuremath{z}}
\newcommand{\character}{\ensuremath{\mathit{char}}}
\newcommand{\identity}{\ensuremath{\mathbf{Id}}}
\newcommand{\scc}{\ensuremath{S}}
\newcommand{\sccs}{\ensuremath{\mathit{SCCs}}}
\newcommand{\decreasing}{\ensuremath{T}}
\newcommand{\witness}{\ensuremath{w}}
\newcommand{\coneOp}{\ensuremath{\mathit{cone}}}
\newcommand{\cone}{\ensuremath{C}}
\newcommand{\difference}{\ensuremath{e}}
\newcommand{\result}{\ensuremath{\mathit{result}}}
\newcommand{\conSysA}{\ensuremath{A}}
\newcommand{\conSysB}{\ensuremath{B}}
\newcommand{\activeP}{\ensuremath{a}}
\newcommand{\activeQ}{\ensuremath{b}}
\newcommand{\primal}{\ensuremath{P}}
\newcommand{\dual}{\ensuremath{Q}}
\newcommand{\timeSAT}{\ensuremath{Z}}
\newcommand{\poly}{\ensuremath{\mathit{poly}}}
\newcommand{\bits}{\ensuremath{\mathit{L}}}
\newcommand{\timeLP}{\ensuremath{\timeSAT'}}
\newcommand{\complexity}{\ensuremath{\mathit{comp}}}
\newcommand{\uInv}{\ensuremath{\mathit{uBound}}}
\newcommand{\compBound}{\ensuremath{\mathit{cBound}}}
\newcommand{\expression}{\ensuremath{\mathit{exp}}}
\newcommand{\cycleVal}{\ensuremath{v}}
\newcommand{\nonPosTrans}{\ensuremath{K}}
\newcommand{\constant}{\ensuremath{p}}
\newcommand{\constantAlt}{\ensuremath{q}}
\newcommand{\parameterizedLP}{\ensuremath{C}}
\newcommand{\rationalLP}{\ensuremath{C}}
\newcommand{\optimal}{\ensuremath{c}}
\newcommand{\denominator}{\ensuremath{m}}
\newcommand{\emptysequence}{\ensuremath{\epsilon}}
\newcommand{\exampleCS}{\ensuremath{\mathit{csys}}}
\newcommand{\exampleProg}{\ensuremath{\mathit{prog}}}

\newcommand{\rank}{\ensuremath{\mathit{rank}}}
\newcommand{\qrank}{\ensuremath{\mathit{qrank}}}
\newcommand{\affineRank}{\ensuremath{\mathit{affine}}}
\newcommand{\sumRank}{\ensuremath{\mathit{sum}}}
\newcommand{\sccRank}{\ensuremath{\mathit{combine}}}
\newcommand{\rankProc}{\ensuremath{\mathit{ranking}}}
\newcommand{\domain}{\ensuremath{W}}
\newcommand{\domainElement}{\ensuremath{a}}

\newcommand{\norm}[1]{\left\lVert #1 \right\rVert}

\section{Introduction}

Vector addition systems (VASs)~\cite{journals/jcss/KarpM69} or equivalently petri nets are one of the most popular models for the representation and the analysis of parallel processes~\cite{journals/eik/EsparzaN94}.
Vector addition systems with states (VASSs)~\cite{journals/tcs/HopcroftP79} enrich VASs with a finite control, which makes them one of the most popular models of concurrent systems as the finite control allows to model different communication primitives, e.g., shared memory~\cite{journals/toplas/0001KW14}.
In this paper we study the problem of analyzing the termination of a given VASS for \emph{all} initial states.
This problem is central in the analysis of \emph{parameterized systems}, where one wants to analyze the behaviour of the system for all instances of the parameter; we discuss examples below.
Kosaraju and Sullivan have shown that termination for all initial states is decidable in PTIME~\cite{conf/stoc/KosarajuS88}.
In this paper we extend the results of Kosaraju and Sullivan in several ways.
First, we present a complete construction of ranking functions for VASSs.
Ranking functions provide valuable information: they represent an explanation for the progress of the VASS, which can be reported to the user of a verification tool, and can be used as certificates for termination.
We derive the ranking function construction as the \emph{dual} (in the sense of linear programming) of the algorithm of Kosaraju and Sullivan, which searches for counterexamples to termination.
Our main algorithm makes use of this duality and simultaneously runs the ranking function construction and search for termination counterexamples in order to achieve completeness.
Second, we apply the obtained ranking functions for the computational complexity analysis of VASSs.
Here, computational complexity refers to the problem of computing for a given VASS $\vass$ the maximal number of steps $\vass$ can take in terms of the size of the initial vector; concretely, we are interested in initial vectors whose components are bounded by some parameter $\parameter \in \mathbb{N}$.
We show how to derive upper and lower complexity bounds from the results computed by our algorithm.
We obtain a precise classification of the asymptotic complexity of conservative VASSs, namely that the complexity is always $\Theta(\parameter^k)$ for an integer $k$ between zero and the dimension of the considered VASS.
This precise characterization gives some evidence that our ranking function construction is optimal.
Finally, we state a method that precisely analyzes VASSs with linear complexity:
given a VASS $\vass$ the method either returns that $\vass$ has at least quadratic complexity or returns that $\vass$ has complexity $\optimal \parameter$ for some $c \in \mathbb{Q}$.

\begin{figure}[t]
\begin{tabular}{l|l|l|l}

\begin{minipage}{3cm}
\begin{tikzpicture}[scale=0.4, node distance = 1cm, auto]
    \node (t1) [xshift=2cm] {$i$};
    \node (t2) [below of=t1,node distance=1.5cm] {$j$};
    \node (t3) [right of=t1,node distance=1cm] {$k$};
    \path(t1) edge [line] node [above]
    {p==ff, p:=tt}(t3)
    (t1) edge [line,bend left] node [right]
    {
    {p:=ff}
    }(t2)
    (t2) edge [line,bend left] node [left]
    {p==tt} (t1)
    ;
   \end{tikzpicture}
   \end{minipage}
&
\begin{minipage}{3cm}
\begin{tikzpicture}[scale=0.4, node distance = 1cm, auto]
    \node (t1) [xshift=2cm] {$\loc_{p=\mathtt{tt}}$};
    \node (t2) [below of=t1,node distance=1.5cm] {$\loc_{p=\mathtt{ff}}$};
    \path(t1) edge [line,bend left] node [right]
    {$\begin{pmatrix}
    -1 \\
    1 \\
    0
    \end{pmatrix}$}(t2)
    (t2) edge [loop below,  every loop/.append style={looseness= 20}] node [right, yshift=0.3cm,xshift=0.1cm]
    {$\begin{pmatrix}
    -1 \\
    1 \\
    0
    \end{pmatrix}$}(t2)
    (t2) edge [line,bend left] node [left]
    {$\begin{pmatrix}
    -1 \\
    0 \\
    1
    \end{pmatrix}$} (t1)
    (t1) edge [loop above,  every loop/.append style={looseness= 20}] node [right, yshift=-0.3cm,xshift=0.1cm]
    {$\begin{pmatrix}
    1 \\
    -1 \\
    0
    \end{pmatrix}$}(t1)
    ;
   \end{tikzpicture}
   \end{minipage}
&
\begin{minipage}{4cm}
 \begin{alltt}
  void main(uint N) \{
    uint i = N, j = N;
\(l\sb{1}\!:\)  while (i > 0) \{
      i--;
      j++;
\(l\sb{2}\!:\)    while (j > 0 && *)
          j--;
  \} \}
\end{alltt}
\end{minipage}
&
\begin{minipage}{2,5cm}
\begin{tikzpicture}[scale=0.4, node distance = 1cm, auto]
    \node (t1) [xshift=2cm] {$\loc_1$};
    \node (t2) [below of=t1,node distance=1.5cm] {$\loc_2$};
    \path(t1) edge [line,bend left] node [right]
    {$\begin{pmatrix}
    -1 \\
    1
    \end{pmatrix}$}(t2)
    (t2) edge [loop below,  every loop/.append style={looseness= 20}] node [right, yshift=0.3cm,xshift=0.1cm]
    {$\begin{pmatrix}
    0 \\
    -1
    \end{pmatrix}$}(t2)
    (t2) edge [line,bend left] node [left]
    {$\begin{pmatrix}
    0 \\
    0
    \end{pmatrix}$} (t1)
    ;
   \end{tikzpicture}
   \end{minipage}
\\
   (a) & (b) & (c) & (d)
   \end{tabular}
   \caption{(a) a process template, (b) VASS $\vass_\exampleCS$, (c) a program, (d) VASS $\vass_\exampleProg$}
   \label{fig:intro}
\end{figure}

\paragraph{Motivation and Illustration of our Results.}
In concurrent systems, the number of system processes is often not known in advance,
and thus the system is designed such that a template process can be instantiated an arbitrary number of times.
The \emph{parameterized verification problem} is to analyze the system for all system sizes~\cite{journals/jacm/GermanS92,journals/tcs/FinkelGRB06,conf/lata/AbdullaDB09,conf/fmcad/JohnKSVW13,journals/sigact/BloemJKKRVW16,conf/lpar/AminofRZ15,conf/icalp/AminofRZS15}.
We illustrate VASSs as models of concurrent systems:
Fig.~\ref{fig:intro} (a) states a process template.
A concurrent system consists of $\parameter$ copies of this process template.
The processes communicate via the boolean variable $p$ (this is an instance of communication via shared memory).
The concurrent system is equivalently represented by the VASS $\vass_\exampleCS$ in Fig.~\ref{fig:intro} (b).
$\vass_\exampleCS$ has two locations $\loc_{p=\mathtt{tt}}$ and $\loc_{p=\mathtt{ff}}$, which represent the state of the shared memory.
$\vass_\exampleCS$ has dimension three in order to represent the number of processes in the local states $i$, $j$ and $k$.
The transitions of $\vass_\exampleCS$ reflect the transitions of the process template, e.g., transition $(-1 \ 1 \ 0)^T$ means that one process moves from state $i$ to $j$.
$\vass_\exampleCS$ is \emph{conservative}, i.e., the sum of the entries in each transition adds up to $0$.
We assume that each process starts in location $i$, i.e., we consider the initial valuation $(\parameter \ 0 \ 0)^T$, which is parameterized in the number of processes.
Our construction (Algorithm~\ref{alg:algorithm}) returns a lexicographic ranking function, e.g.,
\begin{equation*}
  \rank(\loc,
  \begin{pmatrix}
  i \\
  j \\
  k
  \end{pmatrix}) = \left\{
  \begin{array}{cc}
    \langle 2i+2j+1,j \rangle, & \text{ if } \loc = \loc_{p=\mathtt{tt}},\\
    \langle 2i+2j,i \rangle, & \text{ if } \loc = \loc_{p=\mathtt{ff}}.
  \end{array}
  \right.
\end{equation*}
$\rank$ is of \emph{order} $2$, i.e., the ranking function maps VASS states to tuples of length at most two.
We obtain the asymptotic complexity $\Theta(\parameter^2)$ for $\vass_\exampleCS$ from our result that relates the order of the ranking function to the computational complexity (Corollary~\ref{cor:tight-conservative}).

In previous work we have described program analysis frameworks for the automated complexity analysis of imperative programs, which use VASSs (and extensions) as  backend~\cite{conf/cav/SinnZV14,journals/jar/SinnZV17}.
For example, our techniques~\cite{conf/cav/SinnZV14,journals/jar/SinnZV17} allow to abstract the program given in Fig.~\ref{fig:intro} (c) to the VASS $\vass_\exampleProg$ stated in Fig.~\ref{fig:intro} (d).
$\vass_\exampleProg$ has two locations $\loc_1$ and $\loc_2$, which correspond to the loop headers of the program.
$\vass_\exampleProg$ has dimension two in order to represent to the variables $i$ and $j$.
The transitions of $\vass_\exampleProg$ correspond to the increments/decrements of the variables.
We consider the initial valuation $(\parameter \ \parameter)^T$, which is parameterized in the program input $\parameter$.
Our construction (Algorithm~\ref{alg:algorithm}) returns a linear ranking function, e.g.,
\vspace{-0.3cm}
\begin{equation*}
  \rank(\loc,
  \begin{pmatrix}
  i \\
  j
  \end{pmatrix}) = \left\{
  \begin{array}{cc}
    \langle 3i+j \rangle, & \text{ if } \loc = \loc_1,\\
    \langle 3i+j+1 \rangle, & \text{ if } \loc = \loc_2.
  \end{array}
  \right.
\end{equation*}
If we evaluate $\rank$ on the initial state $(\loc_1, (\parameter \ \parameter)^T)$ we obtain the precise number of steps $3\parameter + \parameter = 4\parameter$ that $\vass_\exampleProg$ can take until termination.
Here, we note that our construction is not guaranteed to return this ranking function because the coefficients of the ranking function are obtained as the solution to a linear programming problem (however, our algorithm is guaranteed to return a ranking function of smallest order, here, a linear ranking function of order $1$; linear programming heuristics might be used to find small coefficients).
We present an additional method which is guaranteed to establish the precise asymptotic complexity $4\parameter$ of $\vass_\exampleProg$ (Theorem~\ref{thm:precise-complexity}).

We summarize our contributions:
\vspace{-0.1cm}
\begin{itemize}
  \item We describe a complete algorithm for the construction of ranking functions for VASSs and prove its correctness.
      We improve the analysis of the recursion depth of~\cite{conf/stoc/KosarajuS88};
      with a further optimization, this allows us to answer an open question of~\cite{conf/stoc/KosarajuS88}.
  \item We discuss how our algorithm allows to obtain bounds on the computational complexity of VASSs.
      In particular, we obtain a precise characterization of the complexity of conservative VASSs.
  \item We state a technique for computing the precise asymptotic complexity of linear VASSs.
  \item All our constructions are PTIME and are promising for the use in practical analyzers.
\end{itemize}

\paragraph{Related Work.}
The result by Kosaraju and Sullivan~\cite{conf/stoc/KosarajuS88} is in fact only concerned with the detection of zero-cycles.
However, it is easily seen that their method likewise allows the detection of non-negative cycles.
The existence of non-negative cycles in turn is easily seen to be a complete criterion for the non-termination of VASSs by
Dickson's Lemma~\cite{Dickson1913}.

Celebrated results on VASs include establishing the precise complexity of the termination problem for a \emph{fixed} initial state; this problem can be reduced to the boundedness problem, which has been shown to be EXPSPACE-complete~\cite{Lipton76,journals/tcs/Rackoff78}.
This is a stark contrast to the PTIME complexity of termination for \emph{all} initial states, which is much more promising for practical applicability.
We further note that the termination problem for a fixed initial state for a VASS can be reduced to the same problem for VASs.
However, this is not the case for termination over all initial states; the results in this paper show that the complexity of VASs is always linear (in terms of the initial state).

Our previous work on the complexity analysis of VASSs~\cite{conf/cav/SinnZV14,journals/jar/SinnZV17} also presents ranking function constructions and discusses how to use ranking functions for complexity analysis.
The suggested approaches, however, are not complete and have been designed with the goal of practical applicability.

\section{Preliminaries}
\label{sec:preliminaries}

\paragraph{Basic Notation.}
For a set $X$ we denote with $|X|$ the number of elements of $X$.
We recall some mathematical notation for vectors and matrices which is used throughout the paper.
Let $\field$ be one of the rings $\mathbb{N}$, $\mathbb{Z}$ or $\mathbb{Q}$.
We write $\field^I$ for the set of vectors over $\field$ indexed by some set $I$.
We write $\field^{I\times J}$ for the set of matrices over $\field$ indexed by some sets $I$ and $J$.
In superscripts defining vectors or matrices we will use $n$ as a shorthand for $[1,n]$, e.g., we will write $\field^n$ instead of $\field^{[1,n]}$.
We write $\oneVec$ for the vector which has entry $1$ in every component.
Given $a \in \field^I$, we write $a(i) \in \field$ for the entry at line $i \in I$ of $a$ and $\norm{a} = \max_{i \in I} |a(i)|$ for the maximum absolute value of $a$.
Given $A \in \field^{I\times J}$, we write $A(i)$ for the vector in column $i \in I$ of $A$, $A(i,j) \in \field$ for the entry in column $i \in I$ and row $j \in J$ of $A$ and $\norm{A} = \max_{i \in I,j \in J} |A(i,j)|$ for the maximum absolute value of $A$.
We write $\identity$ for the square matrix which has entries $1$ on the diagonal and $0$ otherwise.

We use the following notation for operations on vectors:
Given $a,b \in \field^I$ we write $a + b \in \field^I$ for component-wise addition, $c \cdot a \in \field^I$ for multiplying every component of $a$ by some $c \in \field$ and $a \ge b$ for component-wise comparison.

We use the following notation for operations on vectors and matrices:
Given $A \in \field^{I\times J}$, $B \in \field^{J \times K}$ and $x \in \field^J$, we write $AB \in \field^{I\times K}$ for the standard matrix multiplication,
$Ax \in \field^I$ for the standard matrix-vector multiplication, $A^T \in \field^{J \times I}$ for the transposed matrix of $A$ and $x^T \in \field^{1 \times J}$ for the transposed vector of $x$.

\paragraph{Vector Addition System with States (VASS).}
  A vector addition with states (VASS) $\vass = (\locations(\vass),\transitions(\vass))$ consists of a finite set of \emph{locations} $\locations(\vass)$ and a finite set of \emph{transitions} $\transitions(\vass)$, where transitions $\transitions(\vass) \subseteq \locations(\vass) \times \mathbb{Z}^\dimension \times \locations(\vass)$ for some $\dimension \in \mathbb{N}$; we call $\dimension$ the \emph{dimension} of $\vass$ and write $\dimOp(\vass) = \dimension$.
  We write $\loc_1 \xrightarrow{\update} \loc_2$ to denote a transition $(\loc_1,\update,\loc_2) \in \transitions(\vass)$;
  we call the vector $\update$ the \emph{update} of transition $\loc_1 \xrightarrow{\update} \loc_2$.

  A (finite or infinite) \emph{path} $\paath$ of $\vass$ is a sequence $\loc_0 \xrightarrow{\update_1} \loc_1 \xrightarrow{\update_2} \cdots$ with $\loc_i \xrightarrow{\update_{i+1}} \loc_{i+1} \in \transitions(\vass)$ for all $i$.
  A \emph{cycle} is a path that has the same start- and end-location.
  A \emph{multi-cycle} is a finite set of cycles (the cycles do not necessarily have to be connected).
  Let $\paath = \loc_0 \xrightarrow{\update_1} \loc_1 \xrightarrow{\update_2} \cdots \loc_k$ be a finite path.
  We say $\paath$ has \emph{length} $k$ and \emph{value} $\valueSum(\paath) = \sum_{i \in [1,k]} \update_i$.
  The value $\valueSum(\multicycle)$ of a multi-cycle $\multicycle$ is the sum of the values of its cycles.
  We call a path $\paath$ \emph{non-negative}, if $\valueSum(\paath) \ge 0$.
  We call a multi-cycle $\multicycle$ \emph{non-negative}, if $\valueSum(\multicycle) \ge 0$.
  $\vass$ is \emph{connected}, if for every pair of locations $\loc,\loc' \in \locations(\vass)$ there is a path from $\loc$ to $\loc'$.

  The set of \emph{valuations} $\Val(\vass) = \mathbb{N}^{\dimOp(\vass)}$ consists of the $\dimOp(\vass)$-tuples over the natural numbers (in this paper $\mathbb{N}$ includes $0$).
  The set of \emph{states} $\states(\vass) = \locations(\vass) \times \Val(\vass)$ consists of pairs of locations and valuations.

  A \emph{step} is a triple $((\loc_1,\val_1), \update, (\loc_2,\val_2)) \in \states(\vass) \times \mathbb{Z}^{\dimOp(\vass)} \times \states(\vass)$ such that $\val_2 = \val_1 + \update$ and $\loc_1 \xrightarrow{\update} \loc_2 \in \transitions(\vass)$.
  We write $(\loc_1,\val_1) \xrightarrow{\update} (\loc_2,\val_2)$ to denote a step $((\loc_1,\val_1), \update, (\loc_2,\val_2))$ of $\vass$.
  A \emph{trace} of $\vass$ is a sequence $(\loc_0,\val_0) \xrightarrow{\update_1} (\loc_1,\val_1) \xrightarrow{\update_2} \cdots$ of steps.
  A VASS $\vass$ is \emph{terminating}, if there is no infinite trace of $\vass$.

  VASS $\vass'$ is a \emph{sub-VASS} of a VASS $\vass$, if $\locations(\vass') \subseteq \locations(\vass)$ and $\transitions(\vass') \subseteq \transitions(\vass)$.
  A \emph{strongly-connected component (SCC)} of a VASS $\vass$ is a maximal sub-VASS $\scc$ of $\vass$ such that $\scc$ is connected (our definition allows \emph{trivial} SCCs, where a SCC $\scc$ is trivial iff $\locations(\scc)$ is singleton and $\transitions(\scc) = \emptyset$).
  We denote the SCCs of $\vass$ by $\sccs(\vass)$.
  We say $\vass$ \emph{can be fully decomposed into SCCs}, if for each transition $\transition \in \transitions(\vass)$ there is an SCC $\scc \in \sccs(\vass)$ such that $\transition \in \transitions(\scc)$.

\section{Ranking Functions}

We first introduce the machinery needed for the definition of ranking functions.

A pair $(\domain, >)$ of a set $\domain$ and a relation $> \ \subseteq \domain \times \domain$ is called a \emph{well-founded domain}, if there is no infinite sequence $\domainElement_0 \domainElement_1 \dots$ of elements $\domainElement_i \in \domain$ with $(\domainElement_i,\domainElement_{i+1}) \in \ >$ for all $i$.
We write $\ge$ for the reflexive closure of $>$, i.e., $\domainElement \ge \domainElement'$ iff $\domainElement > \domainElement'$ or $\domainElement = \domainElement'$.

We denote by $\mathbb{N}^*$ the set of \emph{finite sequences} over $\mathbb{N}$, where $\mathbb{N}^*$ includes the empty sequence $\emptysequence$.
Given two sequences $\langle x_1,\ldots,x_k \rangle, \langle y_1,\ldots,y_l \rangle \in \mathbb{N}^*$, we denote their \emph{concatenation} by $\langle x_1,\ldots,x_k \rangle \oplus \langle y_1,\ldots,y_l \rangle = \langle x_1,\ldots,x_k, y_1,\ldots,y_l \rangle$.
Given two functions $f,g: A \rightarrow \mathbb{N}^*$, we define their \emph{concatenation} $f \oplus g: A \rightarrow \mathbb{N}^*$ by setting $(f \oplus g)(a) = f(a) \oplus g(a)$ for all $a \in A$.
We denote by $\mathbb{N}^{\le k} \subseteq \mathbb{N}^*$ the sequences with \emph{length at most} $k$.
We denote by $(\mathbb{N}^*,>_\lex)$ the \emph{lexicographic order}, where $\langle x_1,\ldots,x_k \rangle >_\lex \langle y_1,\ldots,y_l \rangle$ iff there is an index $1 \le i \le \min\{k,l\}$ such that $x_i > y_i$ and $x_j = y_j$ for all $1 \le j < i$.
We remark that $(\mathbb{N}^*,>_\lex)$ is not well-founded, but that every restriction $(\mathbb{N}^{\le k},>_\lex)$ to sequences with length at most $k$ is well-founded.

\begin{definition}
Let $(\domain, >)$ be a well-founded domain and let $\vass$ be a VASS.
We call a function $\rank: \states(\vass) \rightarrow \domain$ a \emph{ranking function} for $\vass$,
if for every step $\state \xrightarrow{\update} \state'$ of $\vass$ we have $\rank(\state) > \rank(\state')$;
we call $\rank$ a \emph{lexicographic ranking function}, if $(\domain, >) = (\mathbb{N}^{\le k},>_\lex)$ for some $k$, and say that
$\rank$ is of \emph{order} $k$.
We call a function $\qrank: \states(\vass) \rightarrow \domain$ a \emph{quasi-ranking function} for $\vass$,
if for every step $\state \xrightarrow{\update} \state'$ of $\vass$ we have $\qrank(\state) \ge \qrank(\state')$;
we call $\qrank$ \emph{ranking} for step $\state \xrightarrow{\update} \state'$, if $\qrank(\state) >_ \qrank(\state')$.
\end{definition}

We will make use of two (quasi-)ranking function combinators.
Their correctness proofs are straightforward and can be found in the appendix.

\begin{lemma}
\label{lem:sum-ranking}
Let $\decreasing \subseteq \transitions(\vass)$ be a subset of transitions such that for each transition $\transition \in \decreasing$ there is a quasi-ranking function $\qrank_\transition:  \states(\vass) \rightarrow \mathbb{N}$ that is ranking for $\transition$.
Then, $\sumRank((\qrank_\transition)_{\transition \in \decreasing}):  \states(\vass) \rightarrow \mathbb{N}$, defined by $\sumRank((\qrank_\transition)_{\transition \in \decreasing})(\loc,\val) = \sum_{\transition \in \decreasing} \qrank_\transition(\loc,\val)$, is a quasi-ranking function which is ranking for all $\transition \in \decreasing$.
\end{lemma}
\begin{proof}
Let $(\loc_1,\val_1) \xrightarrow{\update} (\loc_2,\val_2)$ be a step of VASS.
By the definition of a step, we have $\loc_1 \xrightarrow{\update} \loc_2 \in \transitions(\vass)$.

Let $\transition \in \decreasing$ be some transition.
By the definition of a quasi-ranking function, we have $\qrank_\transition(\loc_1,\val_1) \ge \qrank_\transition(\loc_2,\val_2)$ (1).
Moreover, the inequality is strict if $\transition = \loc_1 \xrightarrow{\update} \loc_2$ (2).

Because (1) holds for every transition $\transition \in \decreasing$, we have $\sumRank((\qrank_\transition)_{\transition \in \decreasing})(\loc_1,\val_1) =  \sum_{\transition \in \decreasing} \qrank_\transition(\loc_1,\val_1) \ge  \sum_{\transition \in \decreasing} \qrank_\transition(\loc_2,\val_2)  = \sumRank((\qrank_\transition)_{\transition \in \decreasing})(\loc_2,\val_2)$.
Because (2) holds for every transition $\transition \in \decreasing$, the inequality is strict if $\loc_1 \xrightarrow{\update} \loc_2 \in \decreasing$.
\end{proof}

\begin{lemma}
\label{lem:combine-ranking}
Let $\decreasing \subseteq \transitions(\vass)$ be a subset of the transitions of $\vass$ such that $\vass' = (\locations(\vass), \transitions(\vass) \setminus \decreasing)$ can be fully decomposed into SCCs.
Let $\qrank: \states(\vass) \rightarrow \mathbb{N}$ be a quasi-ranking function for $\vass$, which is ranking for each transition in $\decreasing$.
For each SCC $\scc$ of $\vass'$, let $\rank_\scc: \states(\vass) \rightarrow \mathbb{N}^*$ be a lexicographic ranking function $\scc$.
Then, $\sccRank(\qrank,(\rank_\scc)_{\scc \in \sccs(\vass')}): \states(\vass) \rightarrow \mathbb{N}^*$,
defined by $\sccRank(\qrank,(\rank_\scc)_{\scc \in \sccs(\vass')})(\loc,\val) = \qrank(\loc,\val) \oplus \rank_\scc(\loc,\val)$ for $\loc \in \locations(\scc)$, is a lexicographic ranking function for $\vass$.
\end{lemma}
\begin{proof}
Let $(\loc_1,\val_1) \xrightarrow{\update} (\loc_2,\val_2)$ be a step of VASS $\vass$.
By the definition of a step, we have $\loc_1 \xrightarrow{\update} \loc_2 \in \transitions(\vass)$.

If $\loc_1 \xrightarrow{\update} \loc_2 \in \decreasing$, we have $\qrank(\loc_1,\val_1) > \qrank(\loc_2,\val_2)$.
Thus, the first component of $\sccRank(\qrank,(\rank_\scc)_{\scc \in \sccs(\vass')})$ decreases and we get\\
$\sccRank(\qrank,(\rank_\scc)_{\scc \in \sccs(\vass')})(\loc_1,\val_1) > \sccRank(\qrank,(\rank_\scc)_{\scc \in \sccs(\vass')})(\loc_2,\val_2)$.

If $\loc_1 \xrightarrow{\update} \loc_2 \not\in \decreasing$ then, we have $\qrank(\loc_1,\val_1) \ge \qrank(\loc_2,\val_2)$.
By assumption, $\loc_1 \xrightarrow{\update} \loc_2$ belongs to some SCC $\scc$ of $\vass'$.
Hence, we have $\rank_\scc(\loc_1,\val_1) > \rank_\scc(\loc_2,\val_2)$.
Thus, we can conclude $\sccRank(\qrank,(\rank_\scc)_{\scc \in \sccs(\vass')})(\loc_1,\val_1) = \langle \qrank(\loc_1,\val_1), \rank_\scc(\loc_1,\val_1) \rangle  >  \langle \qrank(\loc_2,\val_2), \rank_\scc(\loc_2,\val_2) \rangle =$\\
$\sccRank(\qrank,(\rank_\scc)_{\scc \in \sccs(\vass')})(\loc_2,\val_2)$.
\end{proof}

\section{Witness for Non-Termination}
\label{sec:witness-non-termination}

Next we state a criterion for non-termination of a VASS.

\begin{definition}[Witness for non-termination]
Let $\vass$ be a VASS.
We call a non-negative cycle a \emph{witness for non-termination}.
\end{definition}

In order to show that non-negative cycles can be repeated infinitely often, we make use of the following monotonicity property of VASSs:
\begin{proposition}
\label{prop:monotonicity}
Let $\vass$ be a VASS and let $(\loc_0,\val_0) \xrightarrow{\update_0} (\loc_1,\val_1) \xrightarrow{\update_1} \cdots$ be a trace of $\vass$.
Given a valuation $\val_0' \ge \val_0$, also $(\loc_0,\val_0') \xrightarrow{\update_0} (\loc_1,\val_1') \xrightarrow{\update_1} \cdots$ is a trace of $\vass$.
\end{proposition}
\begin{proof}
It is a straightforward proof by induction to show that $\val_i' \ge \val_i$ for all $i$.
Hence, $\val_i' \ge 0$ for  all $i$ and $(\loc_0,\val_0') \xrightarrow{\update_0} (\loc_1,\val_1') \xrightarrow{\update_1} \cdots$ is a trace of $\vass$.
\end{proof}

Next we state that a non-negative cycle indeed allows to show non-termination:

\begin{lemma}
\label{lem:cycle-non-termination}
Let $\vass$ be a VASS which has a witness for non-termination.
Then, $\vass$ does not terminate.
\end{lemma}
\begin{proof}
  Let $\cycle = \loc_0 \xrightarrow{\update_1} \loc_1 \xrightarrow{\update_2} \cdots \xrightarrow{\update_l} \loc_0$ be a non-negative cycle of $\vass$, i.e, $\valueSum(\cycle) = \sum_{i \in [1,l]} \update_i \ge 0$.
  Let $m = \max_{i \in [1,l]} \norm{\update_i}$ be the maximal absolute value of all update entries.
  We define $\val_0 = ml \cdot \oneVec$ to be the vector with entry $ml$ in every component.
  We now claim that $\vass$ has an infinite trace from state $(\loc_0, \val_0)$.
  We set $\val_k = ml \cdot \oneVec + \sum_{i \in [1,k]} \update_i$ for all $k \in [0,l]$.
  By induction we have $\val_k \ge m(l - k) \cdot \oneVec \ge 0$ for all $k \in [0,l]$.
  Hence, $(\loc_0,\val_0) \xrightarrow{\update_1} (\loc_1,\val_1) \xrightarrow{\update_2} \cdots \xrightarrow{\update_l} (\loc_0,\val_l)$is a trace of $\vass$.
  We set $\val_0' = \val_l$.
  We have $\val_0' = ml \cdot \oneVec + \sum_{i \in [1,l]} \update_i = ml \cdot \oneVec + \valueSum(\cycle) \ge ml \cdot \oneVec = \val_0$.
  By Proposition~\ref{prop:monotonicity}, $\vass$ has a trace $(\loc_0,\val_0') \xrightarrow{\update_1} (\loc_1,\val_1') \xrightarrow{\update_2} \cdots \xrightarrow{\update_l} (\loc_0,\val_0'')$.
  Again, we have  $\val_0'' \ge \val_0'$.
  This argument can be repeated to obtain an infinite trace $(\loc_0,\val_0) \xrightarrow{\update_1} \cdots \xrightarrow{\update_l} (\loc_0,\val_0') \xrightarrow{\update_1} \cdots \xrightarrow{\update_l} (\loc_0,\val_0'') \cdots$.
\end{proof}

We remark that the criterion of a non-negative cycle is complete for non-termination.
Using Dickson's lemma, one can easily show that a non-terminating trace is always of shape $\cdots (\loc_i,\val_i) \xrightarrow{\update_i} \cdots (\loc_j,\val_j) \xrightarrow{\update_j} \cdots$ for some location $\loc_i = \loc_j$ and valuations $\val_i \le \val_j$ and hence contains a non-negative cycle.
The development in this paper represents an alternative proof:
We obtain the existence of non-negative cycles solely by arguments from linear programming. 
\section{Farkas' Lemma and the Constraint System of Kosaraju and Sullivan}
\label{sec:farkas}

We will use the following variant of \emph{Farkas' Lemma}, which states that given matrices $A$,$C$ and vectors $b$,$d$, exactly one of the following statements is true:

\vspace{0.3cm}
\begin{tabular}{|c|c|}
\hline
\begin{minipage}[c]{0.4\linewidth}
constraint system ($A$):

\vspace{0.2cm}
there exists $x$ with

\vspace{0.2cm}
$\begin{array}{rcr}
   Ax & \ge & b \\
   Cx & = & d
 \end{array}$

\end{minipage}
  &
\begin{minipage}[c]{0.5\linewidth}
\vspace{0.2cm}
constraint system ($B$):

\vspace{0.2cm}
there exist $y,z$ with

\vspace{0.2cm}
$\begin{array}{rcr}
   y & \ge & 0 \\
   A^T y + C^T z & = & 0 \\
   b^T y + d^T z & > & 0
 \end{array}$
\vspace{0.2cm}
\end{minipage} \\

\hline
\end{tabular}
\vspace{0.2cm}

We will make use of the following convention throughout the paper:
Let $\vass$ be a VASS.
We define the matrix $\updates \in \mathbb{Z}^{\dimOp(\vass) \times \transitions(\vass)}$ by setting $\updates(\transition) = \update$ for all transitions $\transition = (\loc,\update,\loc') \in \locations(\vass)$.
We define the \emph{oriented incidence matrix} $\flowMatrix \in \mathbb{Z}^{\locations(\vass) \times \transitions(\vass)}$ by setting $\flowMatrix(\loc,\transition) = 1$ resp. $\flowMatrix(\loc,\transition) = -1$, if $\transition = (\loc,\update,\loc')$ resp. $\transition = (\loc',\update,\loc)$ and $\loc' \neq \loc$, and $\flowMatrix(\loc,\transition) = 0$, otherwise.
We note that every column $\transition$ of $\flowMatrix$ either contains exactly one $-1$ and $1$ entry (in case the source and target of transition~$\transition$ are different) or only $0$ entries (in case the source and target of transition~$\transition$ are the same).

We now consider the constraint systems ($\conSysA_\transition$) and ($\conSysB_\transition$) stated below.
Both constraint systems are parameterized by a transition $\transition \in \transitions(\vass)$.
Both constraint systems can be solved in PTIME over the integers, because every rational solution obtained by linear programming over the rationals can scaled to an integer solution by multiplying the solution with the least common multiple of the denominators.
Constraint system ($\conSysA_\transition$) is taken from Kosaraju and Sullivan~\cite{conf/stoc/KosarajuS88} (we note only Equation (\ref{multcycle:eq4}) is parameterized by $\transition$).

\vspace{0.3cm}
\begin{tabular}{|c|c|}
\hline
 {\begin{minipage}[c]{0.4\linewidth}
\vspace{0.2cm}
constraint system ($\conSysA_\transition$):

\vspace{0.2cm}
there exists $\counters \in \mathbb{Z}^{\transitions(\vass)}$ with
\begin{align}
  \updates \counters & \ge 0 \label{multcycle:eq1}\\
  \counters & \ge 0 \label{multcycle:eq2} \\
  \flowMatrix \counters & = 0 \label{multcycle:eq3}\\
  \counters(\transition) & \ge 1 \label{multcycle:eq4}
\end{align}
\vspace{-0.5cm}
\end{minipage}}
  &
{\begin{minipage}[c]{0.5\linewidth}
constraint system ($\conSysB_\transition$):

\vspace{0.2cm}
there exist\\
$\rankCoeff \in \mathbb{Z}^{\dimOp(\vass)},\offsets \in \mathbb{Z}^{\locations(\vass)}$ with
\begin{align}
\rankCoeff & \ge  0 \nonumber\\
\offsets & \ge  0 \nonumber\\
\updates^T \rankCoeff + \flowMatrix^T \offsets & \le 0 \text{ with } < \text{ in line } \transition \label{multcycle:eq5}
\end{align}
\end{minipage}}\\
\hline
\end{tabular}
\vspace{0.3cm}

\begin{lemma}[Cited from~\cite{conf/stoc/KosarajuS88}]
\label{lem:solution-is-multcycle}
There is a solution $\counters \in \mathbb{Z}^{\transitions(\vass)}$  to constraints (\ref{multcycle:eq1})-(\ref{multcycle:eq3}) iff there exists a non-negative multi-cycle with $\counters(\transition)$ instances of transition~$\transition$ for each  $\transition \in \transitions(\vass)$.
\end{lemma}
\begin{proof}
"$\Leftarrow$":
We consider some non-negative multi-cycle $\multicycle$ with $\counters(\transition)$ instances of each transition~$\transition$.
Clearly, $\counters$ is non-negative, i.e., constraint (\ref{multcycle:eq2}) is satisfied.
Because of $\valueSum(\multicycle) \ge 0$ and $\valueSum(\multicycle) = \updates \counters$, constraint (\ref{multcycle:eq1}) is satisfied.
Because $\multicycle$ is a multi-cycle we have that $\counters$ satisfies the \emph{flow constraint} (\ref{multcycle:eq3}), which encodes for every location that the number of incoming transitions equals the number of out-going transitions for every location.

"$\Rightarrow$":
We assume that there exists a solution $\counters$ to constraints (\ref{multcycle:eq1})-(\ref{multcycle:eq3}).
From~(\ref{multcycle:eq2}) we have $\counters \ge 0$.
We now consider the multi-graph which contains $\counters(\transition)$ copies of every transition~$\transition$.
From~(\ref{multcycle:eq3}) we have that the multi-graph is balanced, i.e., the number of incoming edges equals the number of outgoing edges for every location.
It follows that every strongly connected component has an Eulerian cycle.
Each of these Eulerian cycles gives us a cycle in the original VASS.
The union of these cycles is the desired multi-cycle.
The multi-cycle is non-negative by~(\ref{multcycle:eq1}).
\end{proof}

\begin{corollary}
\label{cor:cycle-existence}
If constraints (\ref{multcycle:eq1})-(\ref{multcycle:eq3}) are satisfiable by some $\counters \ge \oneVec$ and VASS $\vass$ is connected, then there exists a non-negative cycle that includes every transition.
\end{corollary}
\begin{proof}
Let $\counters \ge \oneVec$ be a solution to constraints (\ref{multcycle:eq1})-(\ref{multcycle:eq3}).
We consider the multi-cycle obtained from Lemma~\ref{lem:solution-is-multcycle}.
Because of $\counters \ge \oneVec$, the multi-cycle contains each transition at least once.
Because VASS $\vass$ is connected, we can connect the cycles of the multi-cycle into a single cycle.
\end{proof}

We now consider constraint system ($\conSysB_\transition$), which we recognize as the dual of constraint system ($\conSysA_\transition$)
in the following Lemma:

\begin{lemma}
\label{lem:ranking-or-witness}
Exactly one of the constraint systems ($\conSysA_\transition$) and ($\conSysB_\transition$) has a solution.
\end{lemma}
\begin{proof}
We fix some transition~$\transition$.
We denote by $\character_\transition \in \mathbb{Z}^{\locations(\vass)}$ the vector with $\character_\transition(\transition') = 1$, if $\transition' = \transition$, and $\character_\transition(\transition') = 0$, otherwise.
Using this notation we rewrite ($\conSysA_\transition$) to the equivalent constraint system ($\conSysA_\transition'$):

\vspace{0.2cm}
\begin{tabular}{|lc|}
\hline
{\begin{minipage}[r]{0.3\linewidth}
constraint system ($\conSysA_\transition'$):
\vspace{0.8cm}
\end{minipage}}
&
{\begin{minipage}[r]{0.3\linewidth}
\vspace{-0.2cm}
\begin{eqnarray*}
  \begin{pmatrix}
  \updates \\
  \identity
  \end{pmatrix}
  \counters & \ge &
  \begin{pmatrix}
  0 \\
  \character_\transition
  \end{pmatrix}\\
  \flowMatrix \counters & = & 0
\end{eqnarray*}
\vspace{-0.6cm}
\end{minipage}}\\
\hline
\end{tabular}
\vspace{0.2cm}

Using Farkas' Lemma, we see that either ($\conSysA_\transition'$) is satisfiable or the following constraint system ($\conSysB_\transition'$) is satisfiable:

\vspace{0.3cm}
\begin{tabular}{|r|c|}
\hline
{\begin{minipage}[r]{0.47\linewidth}
\vspace{0.2cm}
\hspace{-0.15cm}
constraint system ($\conSysB_\transition'$):
\vspace{-0.7cm}
\begin{eqnarray*}
  \begin{pmatrix}
  \rankCoeff \\
  y
  \end{pmatrix} & \ge & 0 \\
  \begin{pmatrix}
  \updates \\
  \identity
  \end{pmatrix}^T
  \begin{pmatrix}
  \rankCoeff \\
  y
  \end{pmatrix} + \flowMatrix^T \offsets & = & 0 \\
  \begin{pmatrix}
  0 \\
  \character_\transition
  \end{pmatrix}^T
  \begin{pmatrix}
  \rankCoeff \\
  y
  \end{pmatrix} + 0^T \offsets & > & 0
\end{eqnarray*}
\vspace{-0.3cm}
\end{minipage}}
  &
{\begin{minipage}[c]{0.45\linewidth}
constraint system ($\conSysB_\transition'$) simplified:
\begin{eqnarray*}
  \rankCoeff & \ge & 0 \\
  y & \ge & 0 \\
  \updates^T \rankCoeff + y + \flowMatrix^T \offsets & = & 0 \\
  y(\transition) & > & 0
\end{eqnarray*}
\end{minipage}}\\
\hline
\end{tabular}
\vspace{0.3cm}

We observe that solutions of constraint system ($\conSysB_\transition'$) are invariant under shifts of $\offsets$, i.e, if $\rankCoeff$, $y$, $\offsets$ is a solution, then $\rankCoeff$, $y$, $\offsets + c \cdot \oneVec$ is also a solution for all $c \in \mathbb{Z}$  (because every row of $\flowMatrix^T$ either contains exactly one $-1$ and $1$ entry or only $0$ entries).
Hence, we can force $\offsets$ to be non-negative.
We recognize that constraint systems ($\conSysB_\transition'$) and ($\conSysB_\transition$) are equivalent.
\end{proof}

Our interest in the dual constraint system ($\conSysB_\transition$) stems from the fact that we can obtain the building blocks for VASS ranking functions from a solution:

\begin{lemma}
\label{lem:affine-ranking}
Let $(\rankCoeff,\offsets)$ be a solution to constraint system ($\conSysB_\transition$).
Then the function $\affineRank(\rankCoeff,\offsets): \states(\vass) \rightarrow \mathbb{N}$, defined by
$\affineRank(\rankCoeff,\offsets)(\loc,\val) = \rankCoeff^T \val + \offsets(\loc)$,
is a quasi-ranking function for $\vass$.
Moreover, $\affineRank(\rankCoeff,\offsets)$ is ranking for transition~$\transition$.
\end{lemma}
\begin{proof}
Let $(\loc_1,\val_1) \xrightarrow{\update} (\loc_2,\val_2)$ be a step of VASS $\vass$.
By the definition of a step, we have $\val_2 = \val_1 + \update$ and $\loc_1 \xrightarrow{\update} \loc_2 \in \transitions(\vass)$.
We have $\affineRank(\rankCoeff,\offsets)(\loc_2,\val_2) = \rankCoeff^T \val_2 + \offsets(\loc_2) = \rankCoeff^T (\val_1 + \update) + \offsets(\loc_2) =
\rankCoeff^T \val_1 + \rankCoeff^T \update + \offsets(\loc_2) =
\rankCoeff^T \val_1 + \offsets(\loc_1) + \update^T \rankCoeff + \offsets(\loc_2) - \offsets(\loc_1)
\le \rankCoeff^T \val_1 + \offsets(\loc_1) = \affineRank(\rankCoeff,\offsets)(\loc_1,\val_1)$,
where we have the inequality because $(\rankCoeff,\offsets)$ is a solution to the constraints ($\conSysA_\transition$) and thus satisfies Equation~\ref{multcycle:eq5}.
Moreover, the inequality is strict for transition~$\transition$.
\end{proof}

\section{Ranking Function Generation Algorithm}

In Algorithm~\ref{alg:algorithm} we state our main procedure $\rankProc(\vass)$ which either returns a ranking function or a witness for non-termination for a connected VASS $\vass$.
If $\vass$ does not have any transitions, $\rankProc$ immediately returns the trivial ranking function, which maps every state to the empty sequence over $\mathbb{N}^*$.
In the first foreach-loop, $\rankProc$ checks for each transition $\transition \in \transitions(\vass)$ if there is a quasi-ranking function $\qrank_\transition$  which is ranking for $\transition$.
In case there is at least one such quasi-ranking function, we combine them into a single quasi-ranking function $\qrank$.
Then the transitions for which $\qrank$ is ranking are removed and $\rankProc$ is called recursively on the remaining SCCs.
If the recursive calls result in ranking functions $\rank_\scc$ for each SCC $\scc$, then
$\qrank$ and the $\rank_\scc$ are combined into a single ranking function which witnesses the termination of $\vass$.
If at some point during the recursive calls no transition can be removed from the considered sub-VASS of $\vass$, then these transitions can be combined into a non-negative cycle and a witness to non-termination has been found.

\begin{algorithm}[t]
\KwIn{a connected VASS $\vass$}
  \lIf{$\transitions(\vass) = \emptyset$}{
    \Return ``termination witnessed by $\lambda \state \in \states(\vass). \emptysequence$''
  }

$\decreasing$ := $\emptyset$\;

\ForEach{transition $\transition \in \transitions(\vass)$}{
  \If{there is a solution $\rankCoeff,\offsets$ to constraint system ($\conSysB_\transition$)}{
    set $\qrank_\transition$ := $\affineRank(\rankCoeff,\offsets)$\;
    $\decreasing$ := $\decreasing \cup \{\transition\}$\;
  }
  \lElse{
    let $\counters_\transition$ be a solution to constraint system ($\conSysA_\transition$)
  }
}

  \lIf{$\decreasing = \emptyset$}{
    \Return ``non-termination witnessed by $\sum_{\transition \in \transitions(\vass)} \counters_\transition$''
  }

$\qrank := \sumRank((\qrank_\transition)_{\transition \in \decreasing})$\;

let $\vass' := (\locations(\vass), \transitions(\vass) \setminus \decreasing)$\;

\ForEach{SCC $\scc$ of $\vass'$}{
  call $\rankProc(\scc)$\;
  \If{$\rankProc(\scc)$ returned a non-termination witness $\witness$}{
        \Return ``non-termination witnessed by $\witness$''\;
  }
  \lElse{
    let $\rank_\scc$ be the ranking function returned by $\rankProc(\scc)$
  }
}

\Return ``termination witnessed by $\sccRank(\qrank,(\rank_\scc)_{\scc \in \sccs(\vass')})$''\;

\caption{$\rankProc(\vass)$ returns a ranking function or a witness for non-termination.}
\label{alg:algorithm}
\end{algorithm}

In order to analyze the termination of the algorithm we consider the cone of cycle values generated by some VASS $\vass$:

$$
\coneOp(\vass) = \{\updates \counters \mid \flowMatrix \counters = 0, \counters \ge 0\}
$$

As usual we define the dimension $\dimOp(\cone)$ of a cone $\cone$ as the dimension of the smallest vector space containing $\cone$.

\begin{lemma}
\label{lem:recursive-call-dimension}
Let $\vass$ be some VASS such that $\rankProc(\vass)$ leads to a recursive call $\rankProc(\scc)$ for some SCC $\scc$ of $\vass$.
Then $\dimOp(\coneOp(\vass)) > \dimOp(\coneOp(\scc))$.
\end{lemma}
\begin{proof}
  Clearly, $\coneOp(\scc) \subseteq \coneOp(\vass)$ because $\scc$ is a sub-VASS of $\vass$ and hence $\dimOp(\coneOp(\vass)) \ge \dimOp(\coneOp(\scc))$.
  By the assumption that there is a recursive call we have $\decreasing \neq \emptyset$.
  We fix some transition $\transition \in \decreasing$.
  We consider a solution $\rankCoeff,\offsets$ to the constraint system ($\conSysB_\transition$).
  We will show that $\coneOp(\scc)$ is contained in the hyperplane $\{x \mid \rankCoeff^T x = 0\}$ while $\coneOp(\vass)$ is not.
  This is sufficient to infer $\dimOp(\coneOp(\vass)) > \dimOp(\coneOp(\scc))$.

  Because $\rankCoeff,\offsets$ is a solution to ($\conSysB_\transition$) we have $\updates^T \rankCoeff + \flowMatrix^T \offsets \le 0$.
  Let $\difference \ge 0$ be the vector such that $\updates^T \rankCoeff_\transition + \flowMatrix^T \offsets + \difference = 0$ (*).
  We note that $\difference(\transition) > 0$ (because the inequality in line $\transition$ of $\updates^T \rankCoeff + \flowMatrix^T \offsets \le 0$ is strict) and $\difference(\transition') = 0$ for all transitions $\transition' \in \transitions(\scc)$ (otherwise $\transition'$ we would have been added to the $\decreasing$).
  From (*) we get by a variant of Farkas' Lemma that $\forall \counters. \flowMatrix \counters = 0 \Rightarrow  \rankCoeff_\transition^T \updates \counters + \difference^T \counters = 0$ (\#).

  We now consider a vector $\counters \ge 0$ with $\counters(\transition') = 0$ for all transitions $\transition' \in \transitions(\vass) \setminus \transitions(\scc)$.
  With $\difference(\transition') = 0$ for all transitions $\transition' \in \transitions(\scc)$ we get $\difference^T \counters = 0$.
  Thus, we get $\rankCoeff_\transition^T \updates \counters = 0$ from (\#).
  We conclude $\coneOp(\scc) = \{\updates \counters \mid \flowMatrix \counters = 0, \counters \ge 0, \counters(\transition') = 0 \text{ for } \transition' \in \transitions(\vass) \setminus \transitions(\scc) \} \subseteq \{x \mid \rankCoeff_\transition^T x = 0\}$.

  We now consider a vector $\counters \ge 0$ with $\counters(\transition) \ge 1$.
  With $\difference(\transition) > 0$ we get $\difference^T \counters > 0$.
  Thus, we get $\rankCoeff_\transition^T \updates \counters < 0$ from (\#).
  We conclude $\coneOp(\vass) \not\subseteq \{x \mid \rankCoeff_\transition^T x = 0\}$.
\end{proof}

\begin{theorem}
\label{thm:termination-dimension}
$\rankProc(\vass)$ terminates with recursion depth at most $\dimOp(\vass) + 1$.
\end{theorem}
\begin{proof}
By Lemma~\ref{lem:recursive-call-dimension} we have that the dimension of $\coneOp(\vass)$ decreases with every recursive call.
With $\dimOp(\coneOp(\vass)) \le \dimOp(\vass)$, we get that the recursion depth is bounded by $\dimOp(\vass) + 1$ and $\rankProc(\vass)$ is guaranteed to terminate.
\end{proof}

Next we state the correctness of Algorithm~\ref{alg:algorithm}.
The proofs of these results are fairly straightforward and can be found in the appendix.

\begin{theorem}
\label{thm:correctness-non-termination}
If $\rankProc(\vass)$ returns ``non-termination witnessed by $\witness$'',
then $\witness$ is a witness for the non-termination of $\vass$.
\end{theorem}
\begin{proof}
Let $\vass'$ be the VASS of the innermost recursive call which returned $\witness$.
We consider the execution of $\rankProc(\vass')$.
We note that $\transitions(\vass') \neq \emptyset$ and $\decreasing = \emptyset$ (otherwise $\rankProc(\vass')$ would not have returned a witness for non-termination).
Hence, $\rankProc(\vass')$ computed a solution $\counters_\transition$ to ($\conSysA_\transition$) for each transition $\transition \in \transitions(\vass')$.
We note that the solutions of the constraints (\ref{multcycle:eq1})-(\ref{multcycle:eq3}) are closed under addition.
Thus, $\witness = \sum_{\transition \in \transitions(\vass)} \counters_\transition$ is a solution of the constraints (\ref{multcycle:eq1})-(\ref{multcycle:eq3}) and we have $\witness \ge \oneVec$.
As $\vass'$ is connected, we obtain from Corollary~\ref{cor:cycle-existence} that $\vass'$ has a non-negative cycle $\cycle$.
The algorithm maintains the invariant that each VASS of a recursive call is a sub-VASS of the original VASS $\vass$.
Hence, $\cycle$ is a non-negative cycle of $\vass$.
By Lemma~\ref{lem:cycle-non-termination}, $\vass$ does not terminate.
\end{proof}

\begin{lemma}
\label{lem:SCC-decomposotion}
Let $\decreasing \subseteq \transitions(\vass)$ be the set of transitions $\transition \in \transitions$ such that constraint system ($\conSysB_\transition$) is satisfiable.
Let $\vass' = (\locations(\vass), \transitions(\vass) \setminus \decreasing)$ be the VASS without the transitions $\decreasing$.
Then $\vass'$ can be fully decomposed into strongly-connected components.
\end{lemma}
\begin{proof}
We consider a transition $\transition \in \transitions(\vass) \setminus \decreasing$.
Let $\counters_\transition$ be a solution to constraint system ($\conSysA_\transition$).
By Lemma~\ref{lem:solution-is-multcycle} there is a multi-cycle $\multicycle$ with $\counters_\transition(\transition')$ instances of transition~$\transition'$ for each transition $\transition' \in \transitions(\vass)$.
By equation (\ref{multcycle:eq4}) we have that $\counters_\transition(\transition) \ge 1$.
Thus, transition $\transition$ is part of $\multicycle$.
We note that for each transition $\transition'$ of $\multicycle$, we have $\counters_\transition(\transition') \ge 1$ and $\counters_\transition$ is also a solution to constraint system ($\conSysA_{\transition'}$).
Hence, all transitions of $\multicycle$ do not belong to $\decreasing$.
\end{proof}

\begin{theorem}
\label{thm:corectness-termination}
If $\rankProc(\vass)$ returns ``termination witnessed by $\rank$'',
then $\rank$ is a lexicographic ranking function for $\vass$.
Moreover, the order of $\rank$ is bounded by $\dimOp(\coneOp(\vass))$.
\end{theorem}
\begin{proof}
The proof is by induction along the recursion depth.
In case $\transitions(\vass) = \emptyset$, we have $\rank = \lambda \state \in \states(\vass). \emptysequence$;
as $\rank$ is of order $0$ the claim trivially holds.
Otherwise we have $\decreasing \neq \emptyset$.
For each transition $\transition \in \decreasing$, we have from Lemma~\ref{lem:affine-ranking} that $\qrank_\transition = \affineRank(\rankCoeff,\offsets)$ is a quasi-ranking function which is ranking for $\transition$.
From Lemma~\ref{lem:sum-ranking} we have that $\qrank = \sumRank((\qrank_\transition)_{\transition \in \decreasing})$ is a quasi-ranking function which is ranking for every $\transition \in \decreasing$.
For every SCC $\scc$ of $\vass' = (\locations(\vass), \transitions(\vass) \setminus \decreasing)$, we have from the induction assumption that $\rank_\scc = \rankProc(\scc)$ is a ranking function for $\scc$ of order $\dimOp(\coneOp(\scc)) < \dimOp(\coneOp(\vass))$.
By Lemma~\ref{lem:combine-ranking} and Lemma~\ref{lem:SCC-decomposotion}, $\sccRank(\qrank,(\rank_\scc)_{\scc \in \sccs(\vass')})$ is a lexicographic ranking function for $\vass$.
Moreover, $\sccRank(\qrank,(\rank_\scc)_{\scc \in \sccs(\vass')})$ is of order $\max_{\scc \in \sccs(\vass')} \dimOp(\coneOp(\scc)) + 1 \le \dimOp(\coneOp(\vass))$.
\end{proof}

\paragraph{Connectedness of $\vass$.}
The precondition that $\vass$ needs to be connected is required such that solutions to constraints (\ref{multcycle:eq1})-(\ref{multcycle:eq3}) with $\counters \ge \oneVec$ can indeed be turned into cycles (see Corollary~\ref{cor:cycle-existence}).
However, our development also allows the computation of ranking functions for non-connected VASSs.
One needs to decompose $\vass$ into SCCs, compute a ranking functions for each SCC using Algorithm~\ref{alg:algorithm} and compute a reverse topological ordering for the DAG of the SCCs.
Similar to Lemma~\ref{lem:combine-ranking}, the reverse topological ordering can be concatenated with the SCC ranking functions in order to obtain an overall ranking function.

\paragraph{Connection to~\cite{conf/stoc/KosarajuS88}.}
Algorithm~\ref{alg:algorithm} extends algorithm ZCYCLE of Kosaraju and Sullivan~\cite{conf/stoc/KosarajuS88} by ranking function construction.
Algorithm~\ref{alg:algorithm} is a \emph{true extension} in the sense that if one deletes all lines and expression connected to the construction of ranking functions one obtains Algorithm ZCYCLE.
Because of the duality discussed in Section~\ref{sec:farkas},
the ranking function construction part can be interpreted as the \emph{dual} of algorithm ZCYCLE.
Algorithm~\ref{alg:algorithm} makes use of this duality to achieve completeness:
it either returns a ranking function, which witnesses termination, or it returns a non-negative cycle, which witnesses non-termination.
The duality also means that ranking function construction comes essentially for free, as primal-dual LP solvers simultaneously generate solutions for both problems.
An additional result is the improved analysis of the recursion depth:
\cite{conf/stoc/KosarajuS88} uses the fact that the number of locations $|\locations(\vass)|$ is a trivial upper bound of the recursion depth,
while we have shown the bound $\dimOp(\vass)+1$ (see Theorem~\ref{thm:termination-dimension}).
We will use this bound in the complexity analysis of the algorithm.

\paragraph{Complexity analysis of Algorithm~\ref{alg:algorithm}.}
We remark on the complexity of linear programming.
Given $A \in \mathbb{Z}^{n \times m}$ and $b,c \in \mathbb{Z}^m$, interior points methods find a solution to $Ax \le b$ or solve an optimization problem $\max c^Tx$ subject to $Ax \le b$ in $O(n^3L)$, where $L$ is the number of bits needed to encode the problem.

We begin our analysis of Algorithm~\ref{alg:algorithm}.
We denote by $\timeSAT$ the time needed to find a solution to constraint system ($\conSysA_\transition$).
We observe that ($\conSysA_\transition$) consists of $n = |\transitions(\vass)| + 2|\locations(\vass)| + \dimOp(\vass) + 1$ inequalities.
Let $\bits$ be the number of bits needed to encode ($\conSysA_\transition$).
By the above remark, we have $\timeSAT \in O(n^3\bits)$.
Because constraint system ($\conSysB_\transition$) is the dual of constraint system ($\conSysA_\transition$),
computing a solution to ($\conSysB_\transition$) has the same time complexity $\timeSAT$.

The main work of Algorithm~\ref{alg:algorithm}, without considering the recursive calls, happens in the first foreach-loop:
In each iteration of the loop, we consider some transition $\transition \in \transitions(\vass)$ and solve the constraint systems ($\conSysA_\transition$) and ($\conSysB_\transition$).
Thus, the foreach-loop has complexity $O(|\transitions(\vass)|\timeSAT)$.
We now consider the recursive calls.
By Theorem~\ref{thm:termination-dimension}, the recursion depth of $\rankProc(\vass)$ is bounded by $\dimOp(\vass) + 1$ (here we make use of our improved analysis of the recursion depth).
Consider the set of recursive calls made at depth $i$.
The arguments of these calls are all disjoint sub-VASSs of $\vass$.
Thus, the complexity of solving all the constraint systems at level $i$ is bounded by the complexity of solving the constraint systems for $\vass$.
There are $\dimOp(\vass)+1$ recursion levels and the work in each level is bounded by $O(|\transitions(\vass)|\timeSAT)$ as argued above.
Hence, the overall complexity of $\rankProc(\vass)$ is $O(\dimOp(\vass)|\transitions(\vass)|\timeSAT)$, i.e., \emph{polynomial} in the size of $\vass$.

\paragraph{Optimization of Algorithm~\ref{alg:algorithm}.}
The optimization uses the observation that the first foreach-loop of Algorithm~\ref{alg:algorithm}, which requires $|\transitions(\vass)|$ calls to a linear programming solver, can be replaced by a single call.
We make use of a linear objective function to simulate the equation $\counters(\transition) \ge 1$ for each constraint system ($\conSysA_\transition$) in a single optimization problem; similarly, we consider a single optimization problem instead of the  constraint systems ($\conSysB_\transition$).
We state the two optimization problems below:

\vspace{0.3cm}
\begin{tabular}{|c|c|}
\hline
 {\begin{minipage}[c]{0.4\linewidth}
\vspace{0.2cm}
primal problem ($\primal$):

\vspace{0.2cm}
$\max \oneVec^T\activeP$ with
\begin{align*}
  0 & \le \activeP \le \oneVec\\
  \activeP & \le \counters\\
  \updates \counters & \ge 0\\
  \flowMatrix \counters & = 0
\end{align*}
\vspace{-0.5cm}
\end{minipage}}
  &
{\begin{minipage}[c]{0.5\linewidth}
dual problem ($\dual$):

\vspace{0.2cm}
$\min \oneVec^T\activeQ$  with
\begin{align*}
0 & \le \activeQ \le \oneVec\\
\rankCoeff & \ge  0 \\
\offsets & \ge  0 \\
\updates^T \rankCoeff + \flowMatrix^T \offsets & \le -\oneVec + \activeQ
\end{align*}
\end{minipage}}\\
\hline
\end{tabular}
\vspace{0.3cm}

First, we observe that the linear programs (LPs) ($\primal$) and ($\dual$) are satisfiable:
consider $\activeP = 0$ and $\counters = 0$ for ($\primal$), and $\activeQ = \oneVec$, $\rankCoeff = 0$ and $\offsets = 0$ for ($\dual$).
We fix some optimal solution $\activeP,\counters$ of ($\primal$) and some optimal solution $\activeQ,\rankCoeff,\offsets$ of ($\dual$).
We observe that for each $\transition \in \transitions(\vass)$ we have $\activeP(\transition) = 0$ or $\activeP(\transition) = 1$, and also $\activeQ(\transition) = 0$ or $\activeQ(\transition) = 1$ (otherwise the solutions would not be optimal).
By Lemma~\ref{lem:ranking-or-witness}, we have $\activeP(\transition) = \activeQ(\transition)$ for each $\transition \in \transitions(\vass)$.
Hence, $\max \oneVec^T\activeP = \min \oneVec^T\activeQ$; this means that ($\primal$) and ($\dual$) are \emph{dual} LPs in the sense of linear programming.

We set $\decreasing = \{\transition \mid \activeP(\transition) = 0 \} = \{\transition \mid \activeQ(\transition) = 0\}$ and $\rank = \affineRank(\rankCoeff,\offsets)$.
We observe that $\rankCoeff,\offsets$ is a solution to constraint system ($\conSysB_\transition$) for each $\transition \in \decreasing$.
Hence, quasi-ranking function $\rank$ is ranking for each $\transition \in \transitions(\vass)$.
We also observe that $\counters$ is a solution to ($\conSysA_\transition$) for each $\transition \in \transitions(\vass) \setminus \decreasing$ and that we have  $\counters \ge \oneVec$ in case $\decreasing = \emptyset$.
This shows that optimal solutions to ($\primal$) and ($\dual$) can replace the results of the first foreach-loop of Algorithm~\ref{alg:algorithm}.

We now analyze the complexity of the optimized algorithm.
We compare the complexity of solving LP ($\primal$) to the complexity $\timeSAT \in O(n^3\bits)$ of finding a solution to constraint system ($\conSysA_\transition$).
We observe that ($\primal$) consists of $n' = 3|\transitions(\vass)| + 2|\locations(\vass)| + \dimOp(\vass)$ inequalities;
thus, we have $n' \in O(n)$.
The number of bits $\bits'$ needed to encode ($\primal$) is only linearly larger than the number of bits in ($\conSysA_\transition$) as we have added an objective function and an inequality with unit-coefficients; thus, we have $\bits' \in O(\bits)$.
We get that LP ($\primal$) can be solved with the same asymptotic complexity $O(\timeSAT)$.
As before we can be argued that the work at recursion level $i$ is bounded by $O(\timeSAT)$.
Thus, we obtain the overall complexity of $O(\dimOp(\vass)\timeSAT)$.
With this result we affirmatively answer the open question of Kosaraju and Sullivan~\cite{conf/stoc/KosarajuS88}, whether the complexity can be expressed as a polynomial function $\dimOp(\vass)$ times the complexity of a linear program.

\section{Computational Complexity Analysis of VASSs}

In this section we discuss how Algorithm~\ref{alg:algorithm} can be used for analyzing the computational complexity of a given VASS.
We begin with a precise definition of the computational complexity considered in this paper.
Let $\vass$ be a terminating VASS.
We say a state $(\loc,\val) \in \states(\vass)$ is $\parameter$-\emph{bounded}, if $\norm{\val} \le \parameter$ (we recall $\norm{\val} = \max_{i \in [1,\dimOp(\vass)]} |\val(i)|$).
The \emph{computational complexity} $\complexity_\parameter(\vass)$ is the length of the longest trace starting from some $\parameter$-bounded state.

\subsection{Upper Complexity Bounds}

We will make use of the following notions:
Let $\expression: \states(\vass) \rightarrow \mathbb{N}$ be some expression that evaluates a state of a VASS $\vass$ to a natural number.
We define the upper bound invariant $\uInv_\parameter(\expression)$ as the maximum value $\expression(\loc,\val)$ over all states $(\loc,\val)$ reachable from some $\parameter$-bounded state.
Next, we state how lexicographic ranking functions can be used for complexity analysis:

\begin{proposition}
\label{prop:complexity-from-lex-ranking}
  Let  $\rank: \states(\vass) \rightarrow \mathbb{N}^*$ be a lexicographic ranking function with $\rank = \rank_1 \oplus \cdots \oplus \rank_k$ for some $\rank_i: \states(\vass) \rightarrow \mathbb{N}$.
  Then, $\complexity_\parameter(\vass) \le \prod_{i \in [1,k]} \uInv_\parameter(\rank_i)$.
\end{proposition}
\begin{proof}
  The size of the image of $\rank$ restricted to states reachable from some $\parameter$-bounded state is bounded by $\prod_{i \in [1,k]} \uInv_\parameter(\rank_i)$.
  As the value of the ranking function needs to decrease with every step, the claim follows.
\end{proof}

We slightly extend this idea in order to deal with the SCC decomposition of Algorithm~\ref{alg:algorithm}.
Given a VASS $\vass$, for which $\rankProc(\vass)$ returns ``termination'', we define
$$
\begin{array}{ll}
  \compBound_\parameter(\vass) = 1, & \mbox{if } \transitions(\vass) = \emptyset, \\
  \compBound_\parameter(\vass) = \uInv_\parameter(\qrank) \cdot \max_{{\scc \in \sccs(\vass')}} \compBound_\parameter(\scc), & \mbox{otherwise},
\end{array}
$$
where $\qrank$ and $\vass'$ are the values computed by $\rankProc(\vass)$.

By the same argument as for Prop.~\ref{prop:complexity-from-lex-ranking} we obtain the following result:
\begin{proposition}
\label{prop:complexity-from-scc-ranking}
  $\complexity_\parameter(\vass) \le \compBound_\parameter(\vass)$.
\end{proposition}

We apply this result for the complexity analysis of conservative VASSs.
\begin{definition}
A VASS $\vass$ is \emph{conservative},
if for all transitions $\loc \xrightarrow{\update} \loc' \in \transitions(\vass)$ we have $\oneVec^T \update = 0$
\footnote{The results in this paper can be generalized to the following notion of conservative:
For every $\counters \ge 0$ with $\flowMatrix \counters = 0$  we have $\updates \counters = 0$.
This condition can be checked in PTIME.
}.
\end{definition}

We are interested in conservative VASSs because they have trivial upper bound invariants:
Let $\vass$ be a conservative VASS and let $\expression: \states(\vass) \rightarrow \mathbb{N}$ be an expression given by $\expression(\loc,\val) = \rankCoeff^T \val + \offsets(\loc)$ for some $\rankCoeff,\offsets$.
We consider some state $(\loc,\val)$ reachable from some $\parameter$-bounded state.
Because $\vass$ is conservative we have $\expression(\loc,\val) =  \rankCoeff^T \val + \offsets(\loc) \le \norm{\rankCoeff} \oneVec^T \val + \norm{\offsets} \le \norm{\rankCoeff} \dimOp(\vass) \parameter + \norm{\offsets} \in O(\parameter)$.

\begin{theorem}
  \label{thm:conservative-upper-bound}
  Let $\vass$ be a conservative VASS for which $\rankProc(\vass)$ returns a lexicographic ranking function of order $k$.
  Then, $\complexity_\parameter(\vass) \in O(\parameter^k)$.
\end{theorem}
\begin{proof}
  We observe that every quasi-ranking function $\qrank$ encountered during the recursive calls of $\rankProc(\vass)$ has the shape $\qrank(\loc,\val) = \rankCoeff^T \val + \offsets(\loc)$ for some $\rankCoeff,\offsets$ and we have $\uInv_\parameter(\qrank) \in O(\parameter)$ by the above remark.
  The claim follows from Proposition~\ref{prop:complexity-from-scc-ranking} by an induction along the order of the ranking function returned by $\rankProc(\vass)$.
\end{proof}

\subsection{Lower Complexity Bounds}

Our result on lower complexity bounds makes use of the following idea:
We consider an SCC $\scc$ of some VASS $\vass$ for which there is a recursive call $\rankProc(\scc)$ in $\rankProc(\vass)$.
For any trace of $\scc$ there is a non-negative multi-cycle $\multicycle$ of $\vass$ which can be used to ``undo'' the effects of the trace of $\scc$ while only decreasing each vector component of the VASS state linearly.
Using this argument on $k$ nested SCCs one can then obtain the lower bound $\Omega(\parameter^k)$.
The proof is in the appendix for space reasons.

\begin{theorem}
\label{thm:lower-bound}
  Let $\vass$ be a VASS for which $\rankProc(\vass)$ returns a lexicographic ranking function of order $k$.
  Then we have $\complexity_\parameter(\vass) \in \Omega(\parameter^k)$.
\end{theorem}
\begin{proof}
  The proof makes use of the following result from integer linear programming~\cite{journals/jacm/Papadimitriou81}:
  If equation $Ax = b$, given by some matrix $A \in \mathbb{Z}^{m\times n}$ and vector $b \in \mathbb{Z}^m$, is satisfiable by some integer vector $x \ge 0$,
  then there also is a solution $x \in [0,\norm{b} n (m \norm{A})^{2m+1}]^n$.

  We set $n = |\transitions(\vass)| + \dimOp(\vass)$ and $m = |\locations(\vass)| + \dimOp(\vass)$.
  We set $\constant = n (m \norm{\updates})^{2m+1}$ and $\constantAlt = |\transitions(\vass)| \norm{\updates} \constant^2$.
  Let $\cycle$ be some cycle of $\vass$ which visits each node at least once.
  Let $\length$ be the length of $\cycle$.
  We fix some $\parameter$ and set $\parameter' = \frac{1}{\length \norm{\updates} + 1 + |\transitions(\vass)| \norm{\updates} \constantAlt} \parameter$.

  We prove the claim by induction along the order $k$.
  If $k = 0$ then there is nothing to show.
  We assume $k > 0$.
  Thus, we must have $\transitions(\vass) \neq \emptyset$ and there is an SCC $\scc$ of $\vass' = (\locations(\vass), \transitions(\vass) \setminus \decreasing)$ such that the recursive call $\rankProc(\scc)$ returned a ranking function of order $k-1$.

  By induction assumption $\scc$ has a trace of length $\Omega(\parameter'^{k-1})$ for an $\parameter'$-bounded starting state.
  Let $\paath$ be the path associated to this trace.
  Because we are interested only in asymptotic behaviour, we can assume $\paath$ is a cycle.
  Let $\counters_\paath(\transition)$ denote the number of occurrences of transition $\transition$ on $\paath$.
  Let $\cycleVal = \valueSum(\paath) = \updates \counters_\paath$ be the value of $\paath$.
  We have $\counters_\paath \ge 0$, $\flowMatrix \counters_\paath = 0$ and $\updates \counters_\paath \ge -\parameter' \cdot \oneVec$ (because the starting state of the considered trace is $\parameter'$-bounded).

  With an argument from the proof of the theorem cited above~\cite{journals/jacm/Papadimitriou81}, we can obtain some $\counters_\paath^\circ \ge 0$, $\counters_\paath^\diamond \ge 0$ with $\counters_\paath = \counters_\paath^\circ + \counters_\paath^\diamond$, $\updates \counters_\paath^\circ \ge -\parameter' \cdot \oneVec$, $\updates \counters_\paath^\diamond \ge 0$, $\flowMatrix \counters_\paath^\circ = 0$, $\flowMatrix \counters_\paath^\diamond = 0$ and $\counters_\paath^\circ \in [0,\parameter' \constant]^{\transitions(\vass)}$.
  We observe $\norm{\updates \counters_\paath^\circ} \le \parameter' |\transitions(\vass)| \norm{\updates} \constant$.

  We note that for each $\transition \in \transitions(\scc)$ the constraint system ($\conSysA_\transition$) is satisfiable (otherwise we would have  $\transition \in \decreasing$).
  Because solutions of the constraints (\ref{multcycle:eq1})-(\ref{multcycle:eq3}) are closed under addition and multiplication by a positive constant we have that there is a solution $\counters$ to (\ref{multcycle:eq1})-(\ref{multcycle:eq3}) with $\counters(\transition) \ge \counters_\paath^\circ(\transition)$ for all $\transition \in \transitions(\vass)$.
  We set $\countersAlt = \counters - \counters_\paath^\circ$.
  We have $\countersAlt \ge 0$, $\updates \countersAlt \ge -\updates \counters_\paath^\circ$ and $\flowMatrix \countersAlt = 0$.
  Using the cited theorem, there also is a $\hat{\countersAlt} \ge 0$ with $\updates \hat{\countersAlt} \ge - \updates \counters_\paath^\circ$, $\flowMatrix \hat{\countersAlt} = 0$ and $\hat{\countersAlt} \in [0, \parameter' \constantAlt]^{\transitions(\vass)}$.
  We consider the multicycle $M$ associated to such a solution $\hat{\countersAlt}$.

  Let $\loc$ be the first location of $\paath$.
  Because $\cycle$ visits every node at least once we can combine $\cycle$ and the multi-cycle $M$ into a single cycle $\cycle'$ with start and end location $\loc$.
  Let $\val = \parameter \cdot \oneVec$.
  We show that starting from state $(\loc,\val)$ we can $\parameter'$ times
  execute the cycles $\paath$ and $\cycle'$ alternatingly.
  This is sufficient to establish  $\complexity_\parameter(\vass) \in \Omega(\parameter^k)$.

  Let $\val_i$ be the valuation in state $(\loc,\val_i)$ after $0 \le i < \parameter'$ executions of $\paath$ and $\cycle'$.
  We show by induction on $i$ that $\paath$ and $\cycle'$ can be executed one more time.
  We have $\valueSum(\paath) + \valueSum(\cycle') = \updates \counters_\paath^\circ + \updates \counters_\paath^\diamond + \valueSum(M) + \valueSum(\cycle) \ge
  \updates \counters_\paath^\circ + \updates \hat{\countersAlt} +
  \valueSum(\cycle) \ge -\length \norm{\updates} \cdot \oneVec$,
  where we have the last inequality because every step of $\cycle$ decreases each vector component by at most $\norm{\updates}$.
  Hence, we have $\val_i \ge \parameter \cdot \oneVec - i \length \norm{\updates} \cdot \oneVec$.
  In particular we have $\val_i \ge \parameter' \cdot \oneVec$, and $\paath$ can be executed one more time.
  Let $\val_i'$ the valuation after executing $\paath$ in $(\loc,\val_i)$.
  We have $\val_i' \ge \val_i - \parameter' \cdot \oneVec$.
  It remains to show that we can execute $\cycle'$.
  Let $\length'$ be the length of $\cycle'$.
  We have $\length' = \length + \oneVec^T \hat{\countersAlt} \le \length + \parameter' |\transitions(\vass)| \constantAlt$.
  In every step of $\cycle'$ we decrease each vector component by at most $\norm{\updates}$.
  Hence, we need to show $\val_i' \ge \length' \norm{\updates} \cdot \oneVec $.
  Indeed, we have $\val_i' \ge \parameter \cdot \oneVec - i \length \norm{\updates} \cdot \oneVec - \parameter' \cdot \oneVec \ge \norm{\updates} (\length + \parameter' |\transitions(\vass)|  \constantAlt)  \cdot \oneVec$.
  \end{proof}

With Theorem~\ref{thm:conservative-upper-bound} we get the following corollary:

\begin{corollary}
\label{cor:tight-conservative}
  Let $\vass$ be a conservative VASS for which $\rankProc(\vass)$ returns a lexicographic ranking function of order $k$.
  Then, $\complexity_\parameter(\vass) \in \Theta(\parameter^k)$.
\end{corollary}

Corollary~\ref{cor:tight-conservative} gives some evidence that the ranking functions computed by $\rankProc(\vass)$ are optimal: the order of the ranking function returned by $\rankProc(\vass)$ cannot be improved.

We remark on the lower bound $\Omega(\parameter^2)$ for VASS $\vass_\exampleCS$ from the introduction.
This lower bound cannot directly be obtained for the considered initial valuation $(\parameter \ 0 \ 0)^T$.
However, one can establish with a simple pre-analysis that a valuation  $(\Omega(\parameter) \ \Omega(\parameter) \ \Omega(\parameter))^T$ is reachable, and then Theorem~\ref{thm:lower-bound} can be applied. 
\section{Precise Linear Computational Complexity}
\label{sec:linear-complexity}

Our last result is that we can precisely compute the asymptotic computational complexity in case the complexity is linear.
This precisely characterizes the complexity of petri nets, because petri nets correspond to vector addition systems with a single location, and the complexity of such systems is at most linear.

The main idea is as follows:
Given some $\parameter$ we want to solve the optimization problem ($\parameterizedLP_\parameter$) below.
The objection function $\oneVec^T\counters$ counts the number of transitions executed.
The condition $\updates \counters \ge -\parameter$ is satisfied by any trace starting from a $\parameter$-bounded state.
However, $\parameter$ is symbolic and we cannot solve ($\parameterizedLP_\parameter$) directly.
Instead, we solve the problem ($\rationalLP$) over the rationals.
Intuitively we set $\countersAlt = \frac{1}{\parameter} \cdot \counters$.
Because we are interested in the asymptotic computational complexity we can scale a solution to ($\rationalLP$) to a solution over the integers.
The proof is in the appendix for space reasons.

\vspace{0.3cm}
\begin{tabular}{|c|c|}
\hline
 {\begin{minipage}[c]{0.4\linewidth}
\vspace{0.2cm}
parameterized LP ($\parameterizedLP_\parameter$):

\vspace{0.2cm}
$\max \oneVec^T\counters$ with
\begin{align*}
  \counters & \ge 0\\
  \updates \counters & \ge -\parameter \cdot \oneVec\\
  \flowMatrix \counters & = 0
\end{align*}
\vspace{-0.5cm}
\end{minipage}}
  &
{\begin{minipage}[c]{0.5\linewidth}
rational LP ($\rationalLP$):

\vspace{0.2cm}
$\max \oneVec^T\countersAlt$  with
\begin{align*}
  \countersAlt & \ge 0\\
  \updates \countersAlt & \ge -\oneVec\\
  \flowMatrix \countersAlt & = 0
\end{align*}
\end{minipage}}\\
\hline
\end{tabular}
\vspace{0.3cm}

\begin{theorem}
\label{thm:precise-complexity}
  Let $\vass$ be a terminating VASS.
  We consider LP ($\rationalLP$) over the rationals.
  If LP ($\rationalLP$) has a solution $c \in \mathbb{Q}$, then $c \parameter$ is the precise asymptotic computational complexity of $\vass$, i.e., $\complexity_\parameter(\vass) = c \parameter$ for $\parameter \rightarrow \infty$.
  If ($\rationalLP$) is unbounded, then the computational complexity of $\vass$ is at least quadratic.
\end{theorem}
\begin{proof}
  Assume ($\rationalLP$) is unbounded.
  From the theory of linear programming we know that there is a direction in which the polyhedron given by $\countersAlt \ge 0$, $\updates \countersAlt \ge -\oneVec$ and $\flowMatrix \countersAlt = 0$ is unbounded and which increases the objective function $\oneVec^T\countersAlt$.
  Hence, there is a $\countersAlt \ge 0$ with $\updates \countersAlt \ge 0$, $\flowMatrix \countersAlt = 0$ and $\countersAlt(\transition) \ge 1$ for some $\transition \in \transitions(\vass)$.
  Thus, the ranking function returned by $\rankProc(\vass)$ is of order at least $2$ and the complexity is at least quadratic by Theorem~\ref{thm:lower-bound}.

  Assume ($\rationalLP$) is bounded.
  Let $\countersAlt \in \mathbb{Q}^{\transitions(\vass)}$ be an optimal solution.
  We set $\optimal = \oneVec^T\countersAlt$.
  We first show the upper bound.
  We fix some $\parameter$.
  We consider the longest trace starting from some $\parameter$-bounded state.
  Let $\paath$ be the path associated to this trace.
  Because we are interested only in asymptotic behaviour, we can assume $\paath$ is a cycle.
  Let $\counters_\paath(\transition)$ denote the number of occurrences of transition $\transition$ on $\paath$.
  We note that $\updates \counters_\paath = \valueSum(\paath) \ge -\parameter' \cdot \oneVec$ because the starting state of the considered worst-case trace is $\parameter$-bounded.
  Because $\paath$ is a cycle, we have $\flowMatrix \counters_\paath = 0$.
  Hence, $\frac{1}{\parameter} \cdot \counters_\paath$ is a feasible point of LP ($\rationalLP$) and we get $\oneVec^T\frac{1}{\parameter} \cdot \counters_\paath \le \optimal$.
  Thus, $\oneVec^T\counters_\paath \le \optimal\parameter$.
  Because this holds for all $\parameter$, we can conclude $\complexity_\parameter(\vass) \le \optimal\parameter$.

  We show the lower bound.
  We fix some $\parameter$.
  Let $\denominator$ be the least common multiplier of the denominators of $\countersAlt$.
  We set $\counters = \denominator \cdot \countersAlt \in \mathbb{Z}^{\transitions(\vass)}$.
  We have $\counters \ge 0$, $\updates \counters \ge -\denominator \cdot \oneVec$, $\flowMatrix \counters = 0$ and $\oneVec^T\counters = \optimal \denominator$.
  We consider the multicycle $M$ associated to $\counters$.
  Let $\cycle$ be some cycle of $\vass$ which visits each node at least once.
  Let $\length$ be the length of $\cycle$.
  Because $\cycle$ visits every node at least once we can combine $\cycle$ and $\sqrt{\parameter}$ copies of multi-cycle $M$ into a single cycle $\cycle'$.
  Let $\length'$ be the length of $\cycle'$.
  We have $\length' = \length + \sqrt{\parameter}\oneVec^T\counters = \length + \sqrt{\parameter} \optimal \denominator$.
  Let $\loc$ be the start and end location of $\cycle'$.
  We set $\parameter' = \frac{N - (\length + \sqrt{\parameter} \optimal \denominator) \norm{\updates} }{ \norm{\updates} \length + \denominator \sqrt{\parameter}}$ (rounded down if needed).
  Let $\val = \parameter \cdot \oneVec$.
  We show that starting from state $(\loc,\val)$ we can $\parameter'$ times
  execute the cycle $\cycle'$.
  This is sufficient to establish $\complexity_\parameter(\vass) = \optimal \parameter$ because of $\frac{\parameter' \length'}{c\parameter} \rightarrow 1$ for $\parameter \rightarrow \infty$.

  Let $\val_i$ be the valuation in state $(\loc,\val_i)$ after $0 \le i < \parameter'$ executions of $\cycle'$.
  We show by induction on $i$ that $\cycle'$ can be executed one more time.
  We have $\valueSum(\cycle') = \valueSum(\cycle) + \sqrt{\parameter}\valueSum(M) = \valueSum(\cycle) + \sqrt{\parameter} \updates \counters \ge - (\norm{\updates} \length + \denominator \sqrt{\parameter}) \cdot \oneVec$.
  Hence, we have $\val_i \ge \parameter \cdot \oneVec - i (\norm{\updates} \length + \denominator \sqrt{\parameter})\cdot \oneVec$.
  We have to show that we can execute $\cycle'$ one more time.
  In every step of $\cycle'$ we decrease each vector component by at most $\norm{\updates}$.
  Hence, we need to show $\val_i \ge \length' \norm{\updates} \cdot \oneVec$.
  Indeed, we have $\val_i \ge \parameter \cdot \oneVec - i (\norm{\updates} \length + \denominator \sqrt{\parameter}) \cdot \oneVec \ge (\length + \sqrt{\parameter} \optimal \denominator) \norm{\updates} \cdot \oneVec$. 
  \end{proof}
\section{Conclusion and Future Work}

\begin{wrapfigure}[10]{r}{0.21\textwidth}
\vspace{-2cm}
\begin{tikzpicture}[scale=0.4, node distance = 1cm, auto]
    \node (t1) {$\loc_1$};
    \node (t2) [below of=t1,node distance=1.5cm] {$\loc_2$};
    \path(t1) edge [line,bend left] node [right]
    {$\begin{pmatrix}
    0 \\
    0 \\
    0
    \end{pmatrix}$}(t2)
    (t2) edge [loop below,  every loop/.append style={looseness= 20}] node [right, yshift=0.3cm,xshift=0.1cm]
    {$\begin{pmatrix}
    2 \\
    -1 \\
    0
    \end{pmatrix}$}(t2)
    (t2) edge [line,bend left] node [left]
    {$\begin{pmatrix}
    0 \\
    0 \\
    -1
    \end{pmatrix}$} (t1)
    (t1) edge [loop above,  every loop/.append style={looseness= 20}] node [right, yshift=-0.3cm,xshift=0.1cm]
    {$\begin{pmatrix}
    -1 \\
    1 \\
    0
    \end{pmatrix}$}(t1)
    ;
   \end{tikzpicture}
  \vspace{-9mm}
  \caption{VASS $\vass_\mathit{exp}$}
   \label{fig:exponential}
\end{wrapfigure}

The computational complexity is not polynomial in general;
VASS $\vass_\mathit{exp}$ stated in Fig.~\ref{fig:exponential},
which we adapted from an example in~\cite{journals/tcs/HopcroftP79},
has exponential complexity.
There are many interesting problems extending the computational complexity analysis of this paper, for example:
Can we find a precise characterization of VASSs that have polynomial complexity?
Can we obtain precise bounds from ranking functions also for non-conservative VASS?

The optimization problem ($\rationalLP$) from Section~\ref{sec:linear-complexity} can be seen as an optimization problem over a rational relaxation of a VASS~\cite{continuousDA87,conf/icalp/AminofRZS15,conf/lics/BlondinH17}.
We plan to incorporate results from~\cite{conf/lics/BlondinH17} in an extended version of this paper in order to allow initial states with $0$ entries; this would for example allow to consider the initial valuation $(\parameter \ 0)^T$ for the VASS $\vass_\exampleProg$ from the introduction and derive the precise complexity of $3\parameter$ for this initial state. 

\bibliographystyle{plain}
\bibliography{main}

\begin{thebibliography}{10}

\bibitem{conf/lata/AbdullaDB09}
Parosh~Aziz Abdulla, Giorgio Delzanno, and Laurent~Van Begin.
\newblock A language-based comparison of extensions of petri nets with and
  without whole-place operations.
\newblock In {\em {LATA}}, pages 71--82, 2009.

\bibitem{conf/lpar/AminofRZ15}
Benjamin Aminof, Sasha Rubin, and Florian Zuleger.
\newblock On the expressive power of communication primitives in parameterised
  systems.
\newblock In {\em {LPAR}}, pages 313--328, 2015.

\bibitem{conf/icalp/AminofRZS15}
Benjamin Aminof, Sasha Rubin, Florian Zuleger, and Francesco Spegni.
\newblock Liveness of parameterized timed networks.
\newblock In {\em {ICALP}}, pages 375--387, 2015.

\bibitem{journals/sigact/BloemJKKRVW16}
Roderick Bloem, Swen Jacobs, Ayrat Khalimov, Igor Konnov, Sasha Rubin, Helmut
  Veith, and Josef Widder.
\newblock Decidability in parameterized verification.
\newblock {\em {SIGACT} News}, 47(2):53--64, 2016.

\bibitem{conf/lics/BlondinH17}
Michael Blondin and Christoph Haase.
\newblock Logics for continuous reachability in petri nets and vector addition
  systems with states.
\newblock In {\em {LICS}}, pages 1--12, 2017.

\bibitem{continuousDA87}
Rene David and Hassane Alla.
\newblock Continuous petri nets.
\newblock In {\em Proc. 8th European Workshop on Applicationand Theory of Petri
  nets}, pages 275--–294, 1987.

\bibitem{Dickson1913}
Leonard Dickson.
\newblock Finiteness of the odd perfect and primitive abundant numbers with n
  distinct prime factors.
\newblock {\em Am. J. Math}, page 35:413–422, 1913.

\bibitem{journals/eik/EsparzaN94}
Javier Esparza and Mogens Nielsen.
\newblock Decidability issues for petri nets - a survey.
\newblock {\em Elektronische Informationsverarbeitung und Kybernetik},
  30(3):143--160, 1994.

\bibitem{journals/tcs/FinkelGRB06}
Alain Finkel, Gilles Geeraerts, Jean{-}Fran{\c{c}}ois Raskin, and Laurent~Van
  Begin.
\newblock On the \emph{omega}-language expressive power of extended petri nets.
\newblock {\em Theor. Comput. Sci.}, 356(3):374--386, 2006.

\bibitem{journals/jacm/GermanS92}
Steven~M. German and A.~Prasad Sistla.
\newblock Reasoning about systems with many processes.
\newblock {\em J. {ACM}}, 39(3):675--735, 1992.

\bibitem{journals/tcs/HopcroftP79}
John~E. Hopcroft and Jean{-}Jacques Pansiot.
\newblock On the reachability problem for 5-dimensional vector addition
  systems.
\newblock {\em Theor. Comput. Sci.}, 8:135--159, 1979.

\bibitem{conf/fmcad/JohnKSVW13}
Annu John, Igor Konnov, Ulrich Schmid, Helmut Veith, and Josef Widder.
\newblock Parameterized model checking of fault-tolerant distributed algorithms
  by abstraction.
\newblock In {\em {FMCAD}}, pages 201--209, 2013.

\bibitem{journals/toplas/0001KW14}
Alexander Kaiser, Daniel Kroening, and Thomas Wahl.
\newblock A widening approach to multithreaded program verification.
\newblock {\em {ACM} Trans. Program. Lang. Syst.}, 36(4):14:1--14:29, 2014.

\bibitem{journals/jcss/KarpM69}
Richard~M. Karp and Raymond~E. Miller.
\newblock Parallel program schemata.
\newblock {\em J. Comput. Syst. Sci.}, 3(2):147--195, 1969.

\bibitem{conf/stoc/KosarajuS88}
S.~Rao Kosaraju and Gregory~F. Sullivan.
\newblock Detecting cycles in dynamic graphs in polynomial time (preliminary
  version).
\newblock In {\em STOC}, pages 398--406, 1988.

\bibitem{Lipton76}
Richard~J. Lipton.
\newblock {\em The Reachability Problem Requires Exponential space}.
\newblock Research report 62. Department of Computer Science, Yale University,
  1976.

\bibitem{journals/jacm/Papadimitriou81}
Christos~H. Papadimitriou.
\newblock On the complexity of integer programming.
\newblock {\em J. {ACM}}, 28(4):765--768, 1981.

\bibitem{journals/tcs/Rackoff78}
Charles Rackoff.
\newblock The covering and boundedness problems for vector addition systems.
\newblock {\em Theor. Comput. Sci.}, 6:223--231, 1978.

\bibitem{conf/cav/SinnZV14}
Moritz Sinn, Florian Zuleger, and Helmut Veith.
\newblock A simple and scalable static analysis for bound analysis and
  amortized complexity analysis.
\newblock In {\em {CAV}}, pages 745--761, 2014.

\bibitem{journals/jar/SinnZV17}
Moritz Sinn, Florian Zuleger, and Helmut Veith.
\newblock Complexity and resource bound analysis of imperative programs using
  difference constraints.
\newblock {\em J. Autom. Reasoning}, 59(1):3--45, 2017.

\end{thebibliography}

\end{document}